\documentclass{svjour3}                     
\smartqed  

\newcommand{\rr}{\ensuremath{\mathbb{R}}}

\newcommand{\cost}{\text{cost}}

\usepackage{type1cm}
\usepackage{type1ec}

\usepackage{algorithm}
\usepackage{algorithmic}
\usepackage{graphicx}
\usepackage{psfrag}

\usepackage{tikz}
\tikzstyle{place}=[circle,draw=black,fill=gray!15,thick,inner sep=0pt,minimum size=6mm]
\usepackage{enumitem}
\usepackage{tikz}
\usetikzlibrary{arrows}
\usetikzlibrary{decorations.pathmorphing}
\usepackage{subfigure}

\usepackage{amsmath}
\usepackage{amssymb}

\usepackage[normalem]{ulem}
\usepackage{verbatim}

\newtheorem{asum}{Assumption}

\newcommand{\tail}{\ensuremath{\text{tail}}}
\newcommand{\head}{\ensuremath{\text{head}}}
\newcommand{\loc}{\ensuremath{\text{loc}}}
\newcommand{\ca}{\ensuremath{\text{cap}}}

\newcommand{\prob}[3]%
  { \vspace{.25cm}
    \noindent\fbox{\begin{minipage}{.98\textwidth}
      \textsc{#1}

      \vspace{.25cm}
      \begin{center}
        \begin{tabular}{lp{.8\textwidth}}
          Input:  & #2 \\
          Task: & #3
        \end{tabular}
      \end{center}
    \end{minipage}}

    \vspace{.25cm}}%

  { \vspace{.25cm}
    \noindent\fbox{\begin{minipage}{.98\textwidth}
      \textsc{#1}

      \vspace{.25cm}
      \begin{center}
        \begin{tabular}{lp{.8\textwidth}}
          Input:  & #2 \\
          Output: & #3
        \end{tabular}
      \end{center}
   \end{minipage}}
    \begin{enumerate}[(1)]}%
  {\end{enumerate}}

\begin{document}

\title{Flows over time in time-varying networks: \thanks{This work is supported by Berlin Mathematical School and DFG project SK58/7-1.}
}
\subtitle{Optimality conditions and strong duality}

\titlerunning{Optimality conditions and strong duality for MCFP}        

\author{Ronald Koch\and Ebrahim Nasrabadi
}

\authorrunning{R. Koch \and E. Nasrabadi} 

\institute{R.~Koch \at
           Institut f\"{u}r Mathematik, Technische Universit\"{a}t Berlin, Stra\ss e des 17. Juni 136, 10623 Berlin, Germany,
\email{koch@math.tu-berlin.de} \and
E. Nasrabadi\at
Operations Research Center, Massachusetts Institute of Technology, Cambridge, Massachusetts 02139, \email{nasrabad@mit.edu}\\
The work was done when the author was at Institut f\"{u}r Mathematik, Technische Universit\"{a}t Berlin.
}

\date{}

\maketitle

\begin{abstract}

There has been much research on network flows over time due to their important role in real world applications. This has led to many results, but the more challenging continuous time model still lacks some of the key concepts and techniques that are the cornerstones of static network flows. The aim of this paper is to advance the state of the art for dynamic network flows by developing the continuous time analogues of the theory for static network flows. Specifically, we make use of ideas from the static case to establish a reduced cost optimality condition, a negative cycle optimality condition, and a strong duality result for a very general class of network flows over time.

\subclass{90B10\and 49N05 \and 05C38 \and 90C46}
\end{abstract}

\section{Introduction}
In various applications of network flows, flow values on arcs are not constant but may change over time due, e.g., to seasonal altering demands, supplies and arc capacities. Moreover, flow does not travel instantaneously through a network but requires a certain amount of time to travel through each arc. These aspects are captured by \emph{network flows over time} (also called \emph{dynamic network flows}), which were introduced by Ford and Fulkerson~\cite{FordFulkerson58,FordFulkerson62} in the 1950's. Since then, a large number of authors have studied different features of network flow over time models (see \cite{Skutella-Korte09} and the references therein).

Research on flows over time has taken two approaches depending on whether a discrete or continuous representation of time is used \cite{FleischerTardos98}. Research on the discrete-time model typically uses the time-expanded network, either explicitly in the algorithms, or implicitly in order to proof efficiency of theoretically or practically algorithms (see ,e.g., \cite{FleischerSkut-SICOMP}). Research on the continuous-time model has been mostly conducted by Anderson, Philpott and Pullan (see, e.g.,~\cite{Anderson89,AndersonPhilpott89,AndersonPhilpott94,Philpott82,Philpott90,Pullan93,Pullan97}), who consider networks with time-varying parameters. They mainly focus on proving the existence and characterization of optimal solutions, establishing a duality theory, and developing solution algorithms. 

In this paper, we study the \emph{Minimum Cost Flow over time Problem} (hereafter called MCFP for brevity) in the continuous time model. Here the task is to find a minimum cost flow to satisfy (time-varying) demands through a capacitated network, in which arc costs can vary with time, each arc has a transit time, and storage (with a corresponding cost) is allowed at the nodes.  This problem was first introduced by Anderson~\cite{Anderson89}, who characterizes extreme point solutions for the problem given rational transit times. Anderson and Philpott \cite{AndersonPhilpott94} review results relating to MCFP. They introduce a dual problem for MCFP with a corresponding definition of complementary slackness and prove a weak duality result.

A closely related problem to MCFP is the maximum flow over time problem in the continuous-time model. The aim of this problem is to send as much flow as possible  through a capacitated network from a source to a sink within a given time interval. This problem was studied by Anderson \emph{et al.} \cite{AndersonNashPhilpott82}. They introduce the concept of continuous-time cuts and establish a MaxFlow-MinCut theorem (see also \cite{AndersonNash87}) for the case that transit times are zero and the transit capacities are bounded measurable. This result was later extended to arbitrary transit times by Philpott \cite{Philpott90} and to a general model combining both discrete and continuous aspects of flows over time into a single model by Koch \emph{et al.}~\cite{KochNasrSkut11}.

In the absence of transit times, MCFP becomes a special type of {\em Separated Continuous Linear Programs (SCLP)} (see \cite{Anderson78}). Pullan~\cite{Pullan97} considers a more general class of SCLP with time-delays, so-called {\em Separated Continuous
Linear Programs with Time-Delays (SCLPTD)}.  For the case that transit times are rational, Pullan~\cite{Pullan97} transforms SCLPTD into a larger
problem which is very close to a special class of SCLP and
extends some results of SCLP to SCLPTD. In particular, he observes that the main results for SCLP can be extended with ease to give a similar theorem for SCLPTD. 

The common approach in solving CLP as well as SCLP is to approximate the original problem with a finite-dimensional linear program by discretization of time. 
This approach has attracted most of the attention for solving practical problems for the following reasons:
\begin{enumerate}
\item Discretization of time leads to static problems that can be solved by using traditional methods.
\item The discrete approximated solutions converge to an optimal solution for the original problem as the discretization becomes finer.
\end{enumerate}


While discretization-based algorithms are mainly used in practice due to these observations, the size of resulting discrete approximations is enormous, which leads to long computation times. Consequently, a number of authors attempted to generalize the simplex method to solve instances of SCLP as well as MCFP without discretization. In particular, Anderson and Philpott~\cite{AndersonPhilpott89} attempt to develop a simplex algorithm for MCFP with zero transit times and piecewise constant/linear input functions. They discuss how the simplex method can be developed for MCFP to directly produce an exact solution, rather than doing a discretization to get an approximation to the optimal solution. But, there are no guarantees for the convergence of this algorithm and it often produces a sequence of solutions which converge to a suboptimal solution. In the most recent paper, Weiss~\cite{Weiss08} examines SCLP with piecewise linear problem data and develops a simplex algorithm that gives an exact solution after a finite number of iterations. He also characterizes the form of optimal solutions and establishes a strong duality result.

So far we reviewed the literature on network flows over time where transit times are assumed to be time-invariant. There is also a number of models that allow the transit times to vary over time in both discrete-time model \cite{CaiKloksWong97,CaiShaWong01,Miller-HooksPatterson04,Miller-HooksPatterson04,Opasanon06,Miller98,NasrHash10,Wen2013} and continuous-time model \cite{HashNasr11,Orda90,Orda91,OrdaRom95}. In particular, Miller-Hooks and Patterson \cite{Miller-HooksPatterson04} present a pseudo-polynomial time algorithm for the problem of sending a given amount of flow from a single source to a single sink at the shortest possible time in a time-varying network, where all parameters can change at discrete points in time. They also present a technique for converting a network with multiple sources and multiple sinks into an equivalent single source and sink network by adding a small number of nodes and arcs to the existing network. Cai \emph{et al.} \cite{CaiShaWong01} study the problem of sending a given amount of flow from a single source to a single sink with minimum cost in a time-varying network. They present a kind of the successive shortest path algorithm for solving the problem that runs in pseudo-polynomial-time. Nasrabadi and Hashemi \cite{NasrHash10} extent this model to multiple sources and multiple sinks in which the supplies at source nodes and demands at sink nodes may vary with time. They present a discrete-time version of the successive shortest
path algorithm for solving the resulting problem.


\paragraph{Our contribution.}
In this paper, we are concerned with the development of continuous-time analogues to those concepts and techniques which are the cornerstones of static network flows. Specifically, we develop several network based optimality conditions analogous to that found in static network flows for MCFP with piecewise analytic input functions and rational transit times. We derive a strong duality result from these optimality conditions. 

It is worth pointing out that previously, strong duality was developed by Pullan \cite{Pullan96,Pullan97} for SCLP given piecewise analytic problem data and for SCLPTD with rational transit times and piecewise constant/linear input functions. The main result of his paper is a strong duality theorem for SCLP with piecewise analytic data. He first showed that strong duality holds for SCLP under the conditions of piecewise constant/linear problem data in his original paper \cite{Pullan93} as a consequence of his elegant algorithm. This result was the starting point of an extensive duality theory in \cite{Pullan96}. In this paper, we do not follow the Pullan's approach, but rather we make use of ideas from the area of static network flows. 

The remainder of this paper is organized as follows. Section~\ref{sec:preliminaries} provides preliminaries and earlier results that are required for the purpose of the paper. In Section~\ref{sex:SAP}, we introduce the concept of augmenting paths and cycles, and prove the existence of least cost augmenting paths. We then establish a reduced cost optimality condition, a negative cycle optimality condition, and a strong duality result for \mbox{MCFP} in Section~\ref{sec:OptCond}. Section \ref{sec:con} is devoted to our conclusions. Our results in Sections~\ref{sex:SAP} and~\ref{sec:OptCond} are based on the assumption that the cost functions are continuous. In  Appendix~\ref{sec:appendix}, we show that this assumption makes no restriction and all our results can be extended to the case where cost functions have some discontinuities.

\section{Preliminaries}
\label{sec:preliminaries}

In this section, we give a formal description of MCFP and present some existing results, required throughout the paper.

We are given a directed graph $G=(V,E)$ with node set $V$ and arc set $E$ and a time horizon~$T>0$. We denote an arc $e$ from node $v$ to node $w$ by $e=(v,w)$ and assume, without loss of generality, that there is at most one arc between any pair of nodes in~$G$. Each arc~$e$ is associated with two functions; the {\em transit cost}~$c_e:[0,T]\to\rr$ and the {\em transit capacity}~$u_e:[0,T]\to\rr_{\geq 0}$. Here, $c_{e}(\theta)$ gives the cost per flow unit for sending flow in arc~$e$ at time~$\theta$ and $u_{e}(\theta)$ gives an upper bound on the rate (i.e., amount of flow per time unit) at which flow can enter arc~$e$ at time~$\theta$. In addition, for each arc~$e$, a {\em transit time} $\tau_{e}\in\rr$ is given. The transit time is the amount of time required to send flow from the tail to the head of $e$. Thus flow entering arc~$e=(v,w)$ at time~$\theta$ arrives at node~$w$ at time~$\theta+\tau_e$.

Each node~$v$ is associated with three functions; the {\em storage cost}~$c_v:[0,T]\to\rr$, the {\em storage capacity}~$U_v:[0,T]\to\rr_{\geq 0}$, and the {\em supply/demand rate}~${b_v:[0,T]\to\rr}$. Here~$c_v(\theta)$ is the cost per time unit for storing one unit of flow at node~$v$ at time~$\theta$ and~$U_v(\theta)$ is an upper bound on the amount of flow that can be stored at node~$v$ at time~$\theta$. Moreover, depending on whether $b_v(\theta)>0$ or $b_v(\theta)<0$, the value~${|b_v(\theta)|}$ denotes the supply or demand rate of flow at node $v$ at time~$\theta$. Thus, the amount of available supply or required demand up to time $\theta$ at node $v$ equals ${B_v(\theta):=\int_{0}^{\theta}b_v(\vartheta)\,d(\vartheta)}$.  In addition, there may be an initial storage $s_v\in\rr_{\geq 0}$ at node $v$. In this case, the total supply or demand at node $v$ up to time $\theta$ is given by $B_v(\theta):=s_v+\int_{0}^{\theta}b_v(\vartheta)\,d(\vartheta)$.

A \emph{flow over time} (or simply \emph{flow}) $x$ is described by Lebesgue-measurable  functions
\begin{align*}
x_e:[0,T)\longrightarrow \rr_{\geq 0} &&\forall e\in E~.
\end{align*}
The value $x_{e}(\theta)$ denotes the rate of flow entering arc $e$ at the point in time~$\theta$. Therefore, the amount of flow entering arc $e$ up to time $\theta$ equals $X_e(\theta):=\int_{0}^{\theta}x_e(\vartheta)\,d\vartheta$. We let $X_e-\tau$ to be a \emph{shifted} function defined by $(X_e-\tau)(\theta):=X_e(\theta-\tau)$ for each~${\theta\in [0,T)}$. The flow~$x$ induces a storage function $Y_v:[0,T)\longrightarrow \rr$ at node $v$ by the {\em flow conservation constraint}
\begin{align}
\label{eq:FCC}
Y_v:=B_v-\sum_{e\in\delta^+(v)} X_e-\sum_{e\in\delta^-(v)} (X_e-\tau_e)~,
\end{align}
where the value $Y_v(\theta)$ measures the amount of flow stored at node $v$ at time~$\theta$. Here and throughout the paper, the notation $\delta^+(v)$ and $\delta^-(v)$ are used to denote the set of all arcs leaving and entering node $v$, respectively. Obviously, the functions $X$, $B$, $Y$ are absolutely continuous. Moreover, we can assume without loss of generality that $U$ is absolutely continuous as well. That is why we have used a capital letter to denote these functions in order to distinguish them from bounded measurable functions.

The aim of MCFP is to find a flow over time that satisfies all demands and obeys
all transit and storage capacity constraints over the time
interval $[0,T]$, while minimizing the total transit and storage costs.
This problem is formulated as an infinite-dimensional
linear program with a network structure and arc time-delays as below:
\begin{align}
\label{pro:MCFP}
  \tag{MCFP}
  \begin{aligned}
    &\min && \int_{0}^{T}\sum_{e\in E} c_e(\theta) x_e(\theta)\,d\theta+\int_{0}^{T}\sum_{v\in V} c_v(\theta) Y_v(\theta)\,d\theta\\
    &\text{s.t.}
      &&\begin{aligned}[t]
        \sum_{e\in\delta^+(v)} X_e-\sum_{e\in\delta^-(v)} (X_e-\tau_e)+Y_v&= B_v &&\forall v\in V~,\\
        0 \leq x_e &\leq u_e &&\forall e\in E~,\\
        0 \leq Y_v &\leq U_v &&\forall v\in V~.
      \end{aligned}
  \end{aligned}
\end{align}
This formulation is equivalent to that given by Anderson \cite{Anderson89} and provides a very general model for network flow over time problems . 

We say that flow $x$ (with corresponding storage $Y$) is \emph{feasible} if the pair $x$ and $Y$ satisfies the constraints of \eqref{pro:MCFP}.
We require to work within the space $L_\infty([0,T])$ of essentially bounded measurable functions on $[0,T]$ in which equivalent functions differ only on a set of measure zero. In particular, $c_e$, $u_e$, and $x_e$ for all $e\in E$ and $c_v$, $b_v$ for all $v\in V$ belong to $L_\infty([0,T])$. Hence, the {\em feasible region} $\mathcal{F}$ of MCFP is defined as
\begin{align*}
\mathcal{F}&:=\left\{x\in L_\infty^{E}[0,T] \;\big|\; \text{$x$ with corresponding storage $Y$ is feasible for ~\eqref{pro:MCFP}} \right\}.
\end{align*}
We assume that $\mathcal{F}$ is not empty. This guarantees the existence of an optimum solution for~\eqref{pro:MCFP} at an
extreme point of~$\mathcal{F}$ (see~\cite[Theorem 3.1]{Pullan97}).

Anderson and Philpott~\cite{AndersonPhilpott94} develop a dual problem for~\eqref{pro:MCFP} in an analogous manner to that described for static network flows. Before presenting the dual problem, let us assume without loss of generality that the storage costs are zero and there is no initial storage at nodes. Then, the dual problem can be formulated as follows:
\begin{align}
\label{pro:MCFP*}
\tag{MCFP$^*$}
  \begin{aligned}
    &\max && \int_{0}^{T}\sum_{v\in V}b_v(\theta)\pi_v(\theta)\,d\theta-\int_{0}^{T}\sum_{v\in V}{U_v(\theta)\,d\pi_v^{-}(\theta)}+\int_{0}^{T}\sum_{v\in V}{u_e(\theta)\rho_e(\theta)}\,d\theta\\
    &\text{s.t.}
      &&\begin{aligned}[t]
        \pi_{v}(\theta)-\pi_w(\theta+\tau_e)+\rho_{e}(\theta)&\leq c_{e}(\theta) &&\forall e=(v,w)\in E,~\theta\in [0,T)~,\\
   \text{$\pi_v$ of bounded variation and right}&\\
  \text{continuous on $[0,T]$ with $\pi_v(T)$} &=0~&& \forall v\in V~,\\
       \rho_e &\leq 0 &&\forall e\in E~. \\
       \end{aligned}
  \end{aligned}
\end{align}
This formulation requires explanation. Since $\pi_v$ is of
bounded variation, there exist two functions $\pi_v^{+}$ and $\pi_v^{-}$, known as the \emph{Jordan decomposition} of $\pi_v$, that are monotonic increasing on $[0,T]$ with $\pi_v(\theta)=\pi_v^{+}(\theta)-\pi_v^{-}(\theta)$ for $\theta\in [0,T]$ (see, e.g., \cite[Chapter 6 ]{Apostol74}). The functions $\pi_v^{+}$ and $\pi_v^{-}$ are called the positive and negative part of $\pi_v$, respectively. In the objective function, $\pi_v^{-}$ denotes the negative part of $\pi_v$ and the notation $\int_{0}^{T}U_v(\theta)\,d\pi_v^{-}(\theta)$ denotes the Lebesgue-Stieltjes integral of function $U_v$ with respect to function $\pi_v^{-}$ from $0$ to $T$.

For each $e\in E$, the dual variable $\rho_e$ can be eliminated from~\eqref{pro:MCFP*} since it appears in the objective function integrated with $u_e$ which is nonnegative on $[0,T]$. Hence, at an optimum solution $\rho_e$ should be as large as possible. This observation implies that if we know optimal values for the dual variables $\pi$, we can compute the optimal values for $\rho_{e}$ by
\begin{align}
\label{eq:rho}
\rho_{e}(\theta)&=\min\left\{0,c_{e}(\theta)-\pi_v(\theta)+\pi_w(\theta+\tau_e)\right\} &&\forall e=(v,w)\in E,~\theta\in [0,T]~.
\end{align}

\begin{theorem}[Weak Duality, \cite{AndersonPhilpott94}]
\label{thm:WeakDuality} $V[\eqref{pro:MCFP}]\leq V[\eqref{pro:MCFP*}]$.
\end{theorem}
Here and throughout the rest of this paper, we use  $V[OP]$ to denote the optimal value of an optimization problem $OP$. Moreover, we use $V[OP,x]$ to denote the objective function value for a given solution $x$. 

To state the next theorem, we introduce some basic definitions. We say that a monotonic increasing function $f:[a,b]\rightarrow \rr$ is \emph{strictly increasing} at $\theta\in(a,b)$ if~${f(\theta_1)<f(\theta)<f(\theta_2)}$ for any $\theta_1,\theta_2\in [a,b]$ with $\theta\in(\theta_1,\theta_2)$, $f$ is \emph{strictly increasing} at $a$ if $f(a)<f(\theta)$ for every $\theta\in(a,b]$, and $f$ is \emph{strictly increasing} at $b$ if $f(\theta)<f(b)$ for every $\theta\in[a,b)$. A function $f$ of bounded variation on $[a,b]$ is said to be \emph{strictly increasing} at $\theta\in [a,b]$ if $f^{+}$ is strictly increasing at $\theta$, similarly $f$ is \emph{strictly decreasing} at $\theta$ if $f^{-}$ is strictly increasing at $\theta$.

\begin{theorem}[Complementary Slackness, \cite{AndersonPhilpott94}]
\label{thm:CS} Suppose that $x$  with corresponding
storage $Y$ derived from~\eqref{eq:FCC} is feasible for \eqref{pro:MCFP} and $\pi$ with
corresponding~$\rho$ given by~\eqref{eq:rho} is feasible for \eqref{pro:MCFP*}. If for each $e=(v,w)\in E$, $v\in V$, and~${\theta\in [0,T]}$, the following conditions are met:
\begin{enumerate}[label = (CS\arabic*), leftmargin = *]
  \item\label{it:CS1} if $c_{e}(\theta)-\pi_v(\theta)+\pi_w(\theta+\tau_e)>0$, then
  $x_e(\theta)=0$;
  \item\label{it:CS2} if $c_{e}(\theta)-\pi_v(\theta)+\pi_w(\theta+\tau_e)<0$, then $x_e(\theta)=u_{e}(\theta)$;
  \item\label{it:CS3} if $\pi_v$ is strictly increasing at $\theta$, then
  $Y_v(\theta)=0$;
  \item\label{it:CS4} if $\pi_v$ is strictly decreasing at $\theta$, then
  $Y_v(\theta)=U_v(\theta)$.
\end{enumerate}
then $x$ and $\pi$ are optimal for \eqref{pro:MCFP} and \eqref{pro:MCFP*}, respectively.
\end{theorem}
We shall refer to conditions \ref{it:CS1}--\ref{it:CS4} as \emph{complementary slackness} conditions.
The function $\pi$ of bounded variation is said to be {\em complementary slack} with $x$ if $x$ and $\pi$ satisfies these conditions.

It is an open question as to whether a strong duality result can be established whereby $V[\eqref{pro:MCFP}]= V[\eqref{pro:MCFP*}]$ and these values are attained in each program. As noted previously, a primal optimal solution exists for \eqref{pro:MCFP}. Thus we are left with the task to find a dual feasible solution~$\pi$ for which $V[\eqref{pro:MCFP},x]=V[\eqref{pro:MCFP*},\pi]$. In general, strong duality may not hold, even for the special case that all transit times are zero (see \cite{Pullan97MMS} for some examples). However, we show that strong duality can be derived for \eqref{pro:MCFP} and \eqref{pro:MCFP*} under the following assumptions.

\begin{asum}
\label{asum:rational}
The transit times $\tau_e, e\in E$ are all rational as well as the time horizon~$T$.
\end{asum}

\begin{asum}
\label{asum:analytic}
The input functions $c_e,u_e$ for all $e\in E$ and $b_v, u_v$ for each $v\in V$ are all piecewise
analytic on $[0,T]$.
\end{asum}

We recall that a function $h:[0,T]\rightarrow \rr$ is said to \emph{piecewise analytic} if there exists a partition $\{\theta_0,\theta_1,\ldots,\theta_m\}$ of $[0,T]$, $\epsilon>0$, and $h_k$ analytic on $(\theta_{k-1}-\epsilon,\theta_k)$ with~$h_k(\theta)=h(\theta)$ for all $\theta\in [\theta_{k-1},\theta_k)$ and all $k=1,\ldots,m$. It follows from this definition that piecewise analytic functions are right-continuous but not necessarily left-continuous. In fact, they may be discontinuous at a finite number of points.

We suppose that Assumptions~\ref{asum:rational} and \ref{asum:analytic} hold throughout the rest of the paper. These assumptions guarantee the existence of a piecewise analytic optimal solution for \eqref{pro:MCFP}.

\begin{theorem} [Pullan~\cite{Pullan97}]
\label{thm:analytic-flow}
Under Assumptions~\ref{asum:rational} and \ref{asum:analytic}, if $\mathcal{F}$ is nonempty, then \eqref{pro:MCFP} has an optimal solution which is also piecewise analytic on $[0,T]$.
\end{theorem}

\section{Least Cost Augmenting Paths}
\label{sex:SAP}
A key observation in establishing strong duality for the \emph{static} minimum cost flow problem is the fact that starting from some feasible flow we can construct a dual solution if the network contains no augmenting cycles with negative cost. More precisely, the shortest distance labels from one specified node to the other nodes in the residual network define a dual feasible solution which is complementary slack with the given feasible flow. We wish to derive duality results for \eqref{pro:MCFP} along the same lines. To do so, we need to carry over the concept of least cost augmenting paths to the continuous-time setting. 

For each arc $e=(v,w)\in E$, we create a \emph{backward arc} $\overleftarrow{e} :=(w,v)$. This causes no notational conflict since $(w,v)$ is not in $E$ if $(v,w)$ belongs to $E$ due to our assumption that  there is at most one arc between any pair of nodes. 
With each backward arc $\overleftarrow{e}$ we associate a transit time~$\tau_{\overleftarrow{e}}:=-\tau_e$ and a cost function $c_{\overleftarrow{e}}:=-(c_{e}-\tau_e)$. We denote the set of all backward arcs by $\overleftarrow{E}$ and set
$E^\text r := E \cup \overleftarrow{E}$.


Following Philpott \cite{Philpott85}, a \emph{node-time pair} (NTP) is a member of $V\times [0,T]$ which refers to a particular node at a specific time. We say that NTP~$(v,\alpha)$ is \emph{arc-linked} to NTP $(w,\beta)$ if we can arrive at $w$ at time $\beta$ when departing from $v$ at time $\alpha$ via~${e:=(v,w)}$, i.e., $e\in E^\text r$ and~$\beta=\alpha+\tau_{e}$. We also say that NTP $(v,\alpha)$ is \emph{node-linked} to NTP $(w,\beta)$ if $v=w$. 
A \emph{continuous-time dynamic walk} from NTP $(v,\alpha)$ to NTP $(w,\beta)$ is a sequence of NTPs
\begin{align*}
P:(v,\alpha)=(v_1,\theta_{1}),(v_2,\theta_{2}),\ldots,(v_q,\theta_{q})=(w,\beta),
\end{align*}
such that consecutive members are either arc- or node-linked. 
The sequence $P$ is called a \emph{continuous-time dynamic path} if all NTPs are distinct and is called a \emph{continuous-time dynamic cycle} if $q \geq 3$, $(v,\alpha)=(w,\beta)$, and all other NTPs are distinct. For reasons of brevity, hereafter, the term ``continuous-time dynamic'' is omitted when referring to a continuous-time dynamic walk, path, or cycle.

Given a flow over time $x$ with corresponding storage $Y$, we define the corresponding \emph{residual network} as follows. Due to Theorem \ref{thm:analytic-flow}, the flow $x$ is assumed to be piecewise analytic.  For each arc $e \in E$ we define the \emph{residual capacity} of $e$ and $\overleftarrow{e}$ by $u^\text r_e := u^\text r_e - x_e$ and $u^\text r_{\overleftarrow{e}} := x_e - \tau_e$, respectively. Further, for each node $v \in V$, we define the \emph{upper} and \emph{lower residual capacity} of $v$ as $U^\text r_v := U_v - Y_v$ and $L^\text r_v:= Y_v$, respectively.  For any arc $e\in E$, the residual capacities $u^\text r_e$ and $u^\text r_{\overleftarrow{e}}$ give the maximum additional flow rate that can be sent or removed from arc $e$, respectively, without violating the arc capacity constraint $0\leq x_e\leq u_e$.  Similarly, for any node $v\in V$, the residual capacities $U^\text r_v$ and $L^\text r_v$ give the maximum additional flow that can be stored or removed from node $v$, respectively, without violating the node capacity constraint $0\leq Y_v\leq U_v$.

The concept of residual network is based on the following intuitive idea. Suppose that the flow rate into arc $e=(v,w)$ at time $\theta$ is $x_e(\theta)$. Then we can send an additional flow at rate $u_e(\theta)-x_e(\theta)$ into arc $e$ at time $\theta$. Also notice that we can send a flow at rate
$x_e(\theta)$ from node $w$ to node $v$ over the backward arc $\overleftarrow{e}$ at time $\theta+\tau_e$, which amounts to
canceling the existing flow on the arc $e$ at time $\theta$. Whereas sending one unit flow along arc $e$ increases the flow cost by $c_{e}(\theta)$ units, sending one unit flow from node $w$ to
node $v$ on the arc $\overleftarrow{e}$ at time $\theta+\tau_e$ decreases the flow cost by $c_{e}(\theta)$ units. Hence, the residual capacity and cost of the arc $\overleftarrow{e}$ were defined as $x_e-\tau_e$ and $-(c_e-\tau_e)$, respectively. The concept of upper and lower residual capacities at nodes is based on a similar idea.

Next, we introduce the concept of \emph{strict} augmenting paths and cycles for the continuous-time setting. Given a path (or cycle) $P:(v_1,\theta_{1}),\ldots,(v_q,\theta_{q})$, the \emph{residual capacity}
of $P$ is defined as
\begin{align*}
  \ca(P):=\min\{u_1,\ldots,u_{q-1}\},
\end{align*}
where for $k=1,\ldots,q-1$
\begin{align*}
u_{k}:=\begin{cases}
        u^\text r_e(\theta_k) & \text{if } e:=(v_{k},v_{k+1})\in E^\text r\\
        \min\{U^\text r_{v_{k}}(\theta) \mid \theta_{k}\leq \theta\leq \theta_{k+1}\}
          & \text{if } v_{k}=v_{k+1},~\theta_{k}<\theta_{k+1}\\
        \min\{L^\text r_{v_{k}}(\theta)\mid \theta_{k+1}\leq \theta\leq \theta_{k}\}
          & \text{if } v_{k}=v_{k+1},~\theta_{k+1}<\theta_{k}
       \end{cases}~.
\end{align*}
Note that $\ca(P)$ is a real number and not a function as time is already encoded within $P$. Clearly, if $\ca(P)>0$, one can push additional flow along~$P$ without violating capacity constraints. 
We refer to such a path (or cycle) as a \emph{strict augmenting path} (or \emph{cycle}).

%

The \emph{cost} of $P$ is defined as the sum of the costs of the arcs at the times when they appear along $P$, i.e.,
\begin{align}
\label{eq:cost}
c(P):=\sum_{k:v_k\neq v_{k+1}} c_{v_k,v_{k+1}}(\theta_{k}).
\end{align}
Here, the index $k$ varies from $1$ to $q-1$. Recall that the storage costs are supposed to be zero and therefore the cost of a path or cycle depends only on arc transit costs.  We refer to path $P$ as a \emph{least cost strict augmenting path} if $\ca(P)>0$ and in addition, it has the minimum cost among all strict augmenting paths from NTP $(v,\alpha)$ to NTP $(w,\beta)$.

Before with proceeding our discussion, we give an example to show that it is not enough to define an augmenting path to be a path with positive residual capacity.

\begin{example}\label{ex:noLCsAPath}
 Consider a network with two nodes $s$ and $t$, and one arc $e:=(s,t)$.  We let $T:=1$ and define $c_{e}(\theta):=u_{e}(\theta):=2\theta$ for each $\theta \in [0,1]$. Further, we assume that the transit time of $e$ is zero. Storage of one flow unit is allowed at the nodes $s$ and $t$, i.e., $U_s(\theta) := U_t(\theta) :=1$ for all $\theta\in [0,1]$. The storage costs are supposed to be zero. Finally, there is a supply rate of $\theta$ at node $s$ and a demand rate of $\theta$ at node $t$, i.e., $b_s(\theta) := -b_t(\theta) := \theta$ for each $\theta\in [0,1]$. Hence, there is a unique feasible flow $x$ given by $x_{e}=\theta$ which is, of course, optimal.

We now examine the existence of least cost strict augmenting paths from NTP~$(s,0)$ to NTPs $(t,\theta)$, $\theta\in [0,\theta]$. For each~${\theta\in [0,1]}$, the cost of path $P_\theta:(s,0),(t,0),(t,\theta)$ is zero, but $P_\theta$ is not a strict augmenting path as $\ca(P_\theta)=0$. On the other hand, the path $P^\epsilon_\theta:(s,0),(s,\epsilon),(t,\epsilon),(t,\theta)$ with $0<\epsilon < \theta$, which uses arc $e$ at point in time $\epsilon$, is a strict augmenting path whose cost is $2\epsilon$. Hence,  the cost of path $P^\epsilon_\theta$ becomes smaller as $\epsilon$ approaches to zero. Consequently, there exist no least cost strict augmenting path from $(s,0)$ to $(t,\theta)$ for all~${\theta\in [0,1]}$. 
However, for each $\theta\in [0,1]$, one may consider the path $P_\theta$ as an augmenting path because an additional flow can be sent in the neighborhood of path $P_\theta$. Having defined augmenting paths in this manner,  the path $P_\theta$ is a least cost augmenting path from NTP $(s,0)$ to NTP $(t,\theta)$ for each $\theta\in[0,1]$.
\end{example}

It follows from Example~\ref{ex:noLCsAPath} that we need to define  augmenting paths in an appreciate way. Let $P$ be a given path. We can identify $P$ by a sequence of nodes, say $v_1,v_2,\ldots,v_q$ such that $e_i:=(v_i,v_{i+1})\in E^r, i=1,\ldots,q-1$ together with the arrival and departure times $\alpha_i$ and $\beta_i$, respectively, at node $i$ for $i=1,\ldots,q$, where $\beta_{i+1}=\alpha_{i}+\tau_{e_i}$. Given an $\epsilon>0$, we define the \emph{$\epsilon$-neighborhood} $N(P,\epsilon)$ of $P$ as the set of all paths as $P'$ with the node sequence $v_1,v_2,\ldots,v_q$ together with the arrival and departure times $\alpha'_i$ and $\beta'_i$, respectively, for $i=1,\ldots,q$ such that $|\alpha_i-\alpha'_i|<\epsilon$ and $|\beta_i-\beta'_i|<\epsilon$ for all $i=1,\ldots,q$.
In other words, a path $P'$ is contained in $N(P,\epsilon)$ if and only if $P$ and $P'$ are representable by the same sequence of nodes where all arrival and departure times differ by at most~$\epsilon$. 

We say that a path is \emph{augmenting} if it has a strict augmenting path in its neighborhood. Intuitively, the path $P$ is an augmenting path if we can send an additional flow rate along a path in the neighborhood of $P$. However, we might have some augmenting path with zero residual capacity (see Example~\ref{ex:noLCsAPath}). An augmenting path $P$ from $(v,\alpha)$ to $(w,\beta)$ is said to be a \emph{least cost augmenting path} if it has the minimum cost among all augmenting paths from $(v,\alpha)$ to $(w,\beta)$. 
An augmenting cycle is called a \emph{negative augmenting cycle} if its cost is negative.

In the following, we investigate the existence of least cost augmenting paths from a NTP to all other NTPs. 
We consider two arbitrary nodes $s$ and $t$ in $V$ and fix NTP $(s,0)$ as the \emph{source} and NTP $(t,T)$ as the \emph{sink}. The question is whether or not there exists a least cost augmenting path from $(s,0)$ to $(t,T)$. A closely related problem is already studied by Koch and Nasrabadi~\cite{KochNasr10}, who discuss the \emph{continuous-time dynamic shortest path problem} with negative transit times. They prove that a dynamic shortest path exists in general if and only if the cost functions are piecewise analytic and transit times are rational. We use the same techniques as in \cite{KochNasr10} to show that under Assumptions \ref{asum:rational} and \ref{asum:analytic} a least cost augmenting path from NTP~$(s,0)$ to NTP $(t,T)$ exists.

We need to give some definitions. Suppose that $P:(v_1,\theta_{1}),(i_2,\theta_{2}),\ldots,(v_q,\theta_{q})$ is an augmenting path from NTP $(s,0)$ to NTP $(t,T)$. The path $P$ is said to be a \emph{local least cost augmenting path} if there exists an
$\epsilon>0$ such that $c(P)\leq c({P'})$ for all augmenting paths $P'$ in the $\epsilon$-neighborhood of $P$.

We next give a characterization of local least cost augmenting paths. Let $Q:(v_\ell,\theta_\ell),\ldots,(v_r,\theta_r)$ be a subsequence of consecutive NTPs in $P$. We refer to $Q$ as an \emph{arc-subpath} of $P$ if any pair of consecutive NTPs in $Q$ are arc-linked, i.e., $(v_{k},v_{k+1})\in E^r$ for $k=\ell,\ldots,r-1$. In this case, $Q$ can be seen as the sequence  $(v_{\ell},v_{\ell+1}),\ldots,$ $(v_{r-1},v_r)$ of arcs in $E^r$ together with starting time~$\theta_\ell$ from node $v_\ell$. If in addition, $v_{\ell-1}= v_{\ell}$ or $v_{\ell}=v_1$ and $v_{r}= v_{r+1}$ or $v_{r}=v_q$, then~$Q$ is called a \emph{maximal arc-subpath} of $P$. 

In what follows, we suppose that $Q$ is an arc-subpath of $P$. For a point in time~$\alpha\in [0,T]$, we define a path $P|_Q(\alpha)$ as
\begin{align}
\label{eq:P(alpha)}
  &P|_Q(\alpha):
    (v_1,\theta_{1}),\ldots,(v_{\ell^*},\theta_{\ell^*}),
      (v_\ell,\alpha_{\ell}),\ldots,(v_r,\alpha_r),
        (v_{r^*},\theta_{r^*}),\ldots,(v_q,\theta_{q})~, \\
  &\label{eq:P(alpha)_1}\begin{aligned}
  &\text{with  }& \ell^*&:=\begin{cases}
                        \ell - 1 & \text{if } v_{\ell-1} = v_{\ell} \\
                        \ell & \text {if }  v_{\ell-1} \neq v_{\ell}
                      \end{cases}
  &&\text{  and  }& r^*&:=\begin{cases}
                        r + 1 & \text{if } v_{r} = v_{r+1} \\
                        r & \text {if }  v_{r} \neq v_{r+1}
                      \end{cases}
  \end{aligned}
\end{align}
where $\alpha_\ell:=\alpha$ and $\alpha_{k+1}:=\alpha_{k}+\tau_{v_{k},v_{k+1}}$ for $k=\ell,\ldots,r-1$. Roughly speaking,~$P|_Q(\alpha)$ is constructed from $P$ by shifting the starting time of arc-subpath $Q$ from $\theta_\ell$ to $\alpha$. 
It is obvious that $P|_Q(\alpha)$ is contained in the $\epsilon$-neighborhood of $P$ for some $\epsilon>0$ if 
$|\theta_\ell - \alpha| < \epsilon$. 
We next determine a maximal interval $[a,b]$, containing $\theta_\ell$ so that the path $P|_Q(\alpha)$ is an augmenting path for every $\alpha\in [a,b]$.  Here and subsequently, by \emph{maximal} we mean with respect to inclusion. To do so, we define a function $f:[0,T]\rightarrow \rr_{\geq0}$ as
\begin{align*}
f(\theta):=\min_{\ell\leq k\leq r-1}\Bigr\{u^r_{e_k}\Bigl(\theta+\sum_{j=\ell}^{k-1}{\tau_{e_j}}\Bigl)\Bigr\},
\end{align*}
where $e_k=(v_k,v_{k+1})$ for $k=\ell,\ldots,r-1$. For each $\theta\in [0,T]$, the value $f(\theta)$ represents the residual capacity of $Q$ when starting at time~$\theta$. Further, we define two more functions $g_{\ell},g_r:[0,T]\rightarrow \rr_{\geq0}$ as
\begin{align*}
  g_{\ell}(\theta)&:=\begin{cases}
           U^r_{v_\ell}(\theta) &\text{if } \theta\geq \theta_{\ell^*},\\
	   L^r_{v_\ell}(\theta) &\text{if } \theta\leq \theta_{\ell^*},
          \end{cases}
  &&\text{ and }&
  g_{r}(\theta)&:=\begin{cases}
           L^r_{v_r}(\theta+\tau_Q)  &\text{if } \theta+\tau_Q\geq \theta_{r^*},\\
	   U^r_{v_r}(\theta+\tau_Q)&\text{if } \theta+\tau_Q\leq \theta_{r^*},
          \end{cases}
\end{align*}
where $\tau_Q$ is the transit time of arc-subpath $Q$, i.e., $
\tau_Q:=\sum_{\ell\leq k\leq r-1}\tau_{e_k}$. The value $g_\ell(\theta)$ gives an upper bound on the amount of flow that can be stored at node $v_{\ell}$ at time $\theta$ if $\theta\geq \theta_{\ell^*}$, and gives an upper bound on the amount of flow that can be removed from the available storage at node $v_{\ell}$ at time $\theta$ if  $\theta\leq \theta_{\ell^*}$.  A similar interpretation holds for the value $g_r(\theta)$.



Since $P$ is an augmenting path, $P|_Q(\alpha)$ is an augmenting path if $g_\ell$ and $g_r$ are strictly positive on $(\min\{\theta_{\ell^*},\alpha\},\max\{\theta_{\ell^*},\alpha\})$ and  $(\min\{\theta_{r^*}-\tau_Q,\alpha\},\max\{\theta_{r^*}-\tau_Q,\alpha\})$, respectively, and $f$ is not identically zero on any neighborhood of $\alpha$. So let $[a_\ell,b_\ell]$ be the maximal interval containing~$\theta_{\ell^*}$ such that~$g_\ell$ is strictly positive on $(a_\ell,b_\ell)\setminus \{\theta_{\ell^*}\}$. Similarly, let $[a_r,b_r]$ be the maximal interval containing $\theta_{r^*}-\tau_Q$ such that $g_r$ is strictly positive on~$(a_r,b_r)\setminus \{\theta_{r^*}-\tau_Q\}$. Finally, let  $[a_f,b_f]$ be the maximal interval containing~$\theta_{\ell}$ such that~$f$ is strictly positive on $(a_\ell,b_\ell)\setminus \{\theta_{\ell}\}$. Then the interval~${[a,b] := [a_\ell,b_\ell] \cap [a_r,b_r] \cap [a_f,b_f]}$ contains all points in time $\alpha$ such that $P|_Q(\alpha)$ remains an augmenting path. Further, $[a,b]$ contains $\theta_\ell$ as $P=P|_Q(\theta_\ell)$ is an augmenting path. Therefore, $[a,b]$ is the desired maximal interval.

So far, we have determined the maximal interval $[a,b]$ for which the path $P|_Q(\alpha)$ is an augmenting path for all $\alpha\in [a,b]$. We now define a cost function $c_Q:[a,b]\rightarrow \rr$ with respect to $Q$ as
\begin{align}
\label{eq:cost-arc-path}
c_Q(\alpha)&:=\sum_{\ell\leq k\leq r-1}c_{e_k}(\alpha_k)~,
\end{align}
where $\alpha_\ell=\alpha$ and $\alpha_{k+1}(\alpha)=\alpha_{k}+\tau_{e_{k}}$ for $k=\ell,\ldots,r-1$ as before. It is straightforward that the cost function $c_Q$ has a local minimum on $[a,b]$ at the point~$\theta_\ell$ if~$P$ is a local least cost augmenting path. Conversely, if for each arc-subpath $Q$ of $P$ the function $c_{Q}$ attains a local minimum at $\theta_\ell$ within the interval $[a,b]$, then $P$ is a local least cost augmenting path. Thus we have established the following lemma. 

\begin{lemma}
\label{lem:local-path}
The path $P$ is a local least cost augmenting path if and only if for each arc-subpath $Q$ of $P$ with starting time $\theta_\ell$ the cost function $c_{Q}$, given by \eqref{eq:cost-arc-path}, has a local minimum at the point $\theta_\ell$.
\end{lemma}

In what follows, let $\mathcal{P}_{\loc}$ be the set of all augmenting paths~$P$ from NTP $(s,0)$ to NTP $(t,T)$ such that for each \emph{maximal} arc-subpath~$Q$ of $P$ with starting time~$\theta$ the function $c_{Q}$, given by \eqref{eq:cost-arc-path}, has a local minimum at $\theta_\ell$ and is not constant on any open neighborhood containing~$\theta_\ell$. Further, we assume that two paths $P_1$ and $P_2$ are identified if they differ only in the starting time $\theta_1$ and~$\theta_2$ ($\theta_1<\theta_2$), respectively, of one common arc-subpaths $Q$ and $c_{Q}$ is constant over $[\theta_1,\theta_2]$. Note that in this case $P_1$ and~$P_2$ have the same cost, i.e., $c(P_1)=c(P_2)$. Then, for each local least cost augmenting path, one augmenting path with the same cost is contained in $\mathcal{P}_{\loc}$. Notice that $\mathcal{P}_{\loc}$ can contain also paths which are not local least cost augmenting. Nevertheless, the following lemma shows that the set of local least cost augmenting paths from NTP $(s,0)$ to NTP $(t,T)$ is finite.

\begin{lemma}
  \label{lem:analytic-nonempty-finite}
  The set $\mathcal{P}_{\loc}$ is finite.
\end{lemma}
\begin{proof}
Due to Assumption~\ref{asum:rational}, we can assume without loss of generality that the transit times are integral. Therefore, each arc $(v,w)\in E^r$ appears at most $T$ times in any arc-subpath of an arbitrary path. In other words, every arc-subpath of any augmenting path contains at most $T\cdot~|E|$ arcs. Consequently, the number of
possible maximal arc-subpaths is bounded by a constant where two arc-subpaths that differ by the starting time are identified.

We now assume by contradiction that the cardinality of $\mathcal{P}_{\loc}$ is infinite. Hence there exists an infinite number of paths in $\mathcal{P}_{\loc}$ all containing the same
maximal arc-subpath~$Q$, but with different starting times. It then follows from Lemma \ref{lem:local-path} that
the cost function $c_Q$, given by \eqref{eq:cost-arc-path}, has an infinite number of local minimum points. This is a contradiction because $c_Q$ is a piecewise analytic function and has only a finite number of local extrema. This establishes the lemma.
\qed\end{proof}

Before we proceed with our discussion, let us make the following assumption.
\begin{asum}
\label{asum:cost}
The cost functions $c_e, e\in E$ are continuous.
\end{asum}
We suppose that this assumption holds throughout the rest of this and the next section. However, all our results hold for the case in which some cost functions have discontinuities by defining the cost of augmenting paths and cycles in a different, but complicated, way. We discuss further details in~Appendix~\ref{sec:appendix}.

We next show that $\mathcal{P}_{\loc}$ contains the least cost augmenting path from NTP $(s,0)$ to NTP $(t,T)$.

\begin{lemma}
\label{lem:path-exists-sink}
Let $P$ be an augmenting path from NTP $(s,0)$ to NTP $(t,T)$. Then there exists an augmenting path $P'\in\mathcal{P}_{\loc}$ with $c(P')\leq c(P)$.
\end{lemma}
\begin{proof}
If $P\in \mathcal{P}_{\loc}$, then we are done. So we consider the case where $P$ is not in~$\mathcal{P}_{\loc}$. In this case we iteratively apply the following procedure to construct an augmenting path~$P'\in \mathcal{P}_{\loc}$  with $c(P')\leq c(P)$.
\begin{enumerate}[label = (\roman*)]
\item\label{it:nonstoping1}
  Let $Q:(v_k,\theta_k),\ldots,(v_r,\theta_r)$ be a maximal arc-subpath of $P$ such that the cost function $c_{Q}$ does not have a local minimum at $\theta_k$ or is constant on an open interval containing $\theta_k$. Notice that such an arc-subpath exists because of the definition of $\mathcal{P}_{\loc}$ and the fact that $P$ is not in $\mathcal{P}_{\loc}$. Further, choose $P'$ such that it contains a minimal number of arcs.
\item
  Because of Assumption~\ref{asum:cost}, the function $c_{Q}$ is continuous. Thus it takes its minimum over $[a,b]$ at some point~$\theta$. If it has several local minimum, then choose $\theta$ to be the one with maximum value.
\item
  Let $P|_{Q}(\theta)$ be the augmenting path from NTP $(0,1)$ to NTP $(t,T)$ obtained from $P$ by shifting the arc-subpath $Q$ by $\theta_k-\theta$ time units. Since $P|_{Q}(\theta)$ may contain augmenting cycles, we delete all of them in $P|_{Q}(\theta)$.
\item\label{it:nonstoping2} Set $P:=P|_{Q}(\theta)$. If $P$ is not in $\mathcal{P}_{\loc}$, then go to \ref{it:nonstoping1}.
\end{enumerate}
The above procedure terminates after a finite number of iterations and the resulting augmenting path $P$ is contained in $\mathcal{P}_{\loc}$. Further, in each iteration the cost of $P$ does not increase which proves the lemma.
\qed\end{proof}

As a consequent of Lemmas \ref{lem:analytic-nonempty-finite} and \ref{lem:path-exists-sink}, we obtain the following result:
\begin{theorem}
\label{thm:CDSP-existence}
There exists a least cost augmenting path from NTP $(s,0)$ to NTP $(t,T)$. In particular, an augmenting path in $\mathcal{P}_{\loc}$ with minimum cost is the desired path.
\end{theorem}

We can assume that the network $G$ contains an augmenting path from NTP $(s,0)$ to every other NTP. This assumption imposes no loss of generality since this is satisfied, if necessary, by introducing an artificial  node $a$ with infinite storage capacity and zero storage cost, and adding artificial arc $(s,a)$ joining node $s$ to node $a$ and artificial arcs $(a,v), v\in V\setminus\{s\}$, linking node $a$ to all nodes in $V\setminus\{s\}$. Each artificial arc has a zero transit time, a large cost, and large capacity. It is clear that no artificial arc would appear in a least cost augmenting path from $(s,0)$ to any NTP $(v,\theta)$ unless network $G$ contains no augmenting path from $(s,0)$ to $(v,\theta)$ without artificial arcs. Therefore, Lemma \ref{lem:analytic-nonempty-finite} as well as Lemma \ref{lem:path-exists-sink} remain true if NTP $(t,T)$ is replaced by every other NTP $(v,\theta)$. This leads to the main result of this section.

\begin{theorem}
\label{thm:CDSP-analytic}
Suppose that $x$ is a piecewise analytic solution for \eqref{pro:MCFP}. For each NTP $(v,\theta)$, let $d_v(\theta)$ be the cost of a least cost augmenting path from $(s,0)$ to $(v,\theta)$. Then, for each node $v\in V$, the label $d_v(\theta)$ exists for all $\theta\in [0,T]$ and the function $d_v:[0,T]\rightarrow \rr$ is piecewise analytic.
\end{theorem}

\begin{proof}
 The existence of $d_v(\theta)$ follows from Lemmas~\ref{lem:analytic-nonempty-finite} and \ref{lem:path-exists-sink} for each NTP $(v,\theta)$. It thus remains to show that $d_v$ is piecewise analytic on $[0,T]$ for each $v\in V$.  In the following we fix a node $v\in V$. Similar to the definition of $\mathcal{P}_{\loc}$, define $\mathcal{P}_{\loc}(\theta)$ as the set of augmenting paths~$P$ from $(s,0)$ to $(v,\theta)$ such that for each maximal arc-subpath~$Q$ of $P$ with starting time~$\theta'$ the function $c_{Q}$ has a local minimum at $\theta'$ and is not constant on any open neighborhood containing~$\theta'$. Then~$\mathcal{P}_v:=\cup_{\theta\in [0,T]} \mathcal{P}_{\loc}(\theta)$ contains (nearly) all least cost augmenting paths for any point in time~$\theta \in [0,T]$.

  Next we define an \emph{equivalence relation} $\sim$ on $\mathcal{P}_{v}$. let $P:(v_1,\theta_1),\ldots,(v_q,\theta_q)$ and $P':(v'_1,\theta'_1),\ldots,(v'_{q'},\theta'_{q'})$ be two members of $\mathcal{P}_v$. We define $\sim$ on $\mathcal{P}_{v}$ by $P\sim P'$ if and only if $q=q'$ and there is some $r\in \{1,\ldots,q-1\}$ such that
\begin{enumerate}[label = (\roman*)]
 \item $(v'_k,\theta'_k)=(v'_k,\theta'_k)$ for each $k\leq r$, $v_r=v_{r+1}=v'_{r}=v'_{r+1}$ and $\theta_{r+1}\neq \theta'_{k+1}$,
 \item $v_k=v'_k$ for each $k\geq r+1$ and the NTP sequences $(v_{r+1},\theta_{r+1}),\ldots,(v_q,\theta_q)$ and $(v'_{r+1},\theta'_{r+1}),\ldots,(v'_{q'},\theta'_{q'})$ are arc-subpaths of $P$ and $P'$, respectively.
\end{enumerate}

  Roughly speaking, $P$ and $P'$ are equivalent if they differ only in the starting time of the last maximal arc-subpath. For an equivalence class $[P]$ we denote by $P_1$ the path consisting of the first $r$ NTPs of~$P$ and by $P_2$ the arc-path consisting of the last $q-r+1$ NTPs of $P$. Note that $P_1$ and $P_2$ can be the empty path: if~$P$ is an arc-path we put it in the equivalence class for which $P_1=\emptyset$ and $P_2=P$ and if the the last tow NTPs in $P$ are node-linked, we put it in the equivalence class for which $P_1=P$ and $P_2=\emptyset$. Therefore, each equivalence class $[P]$ is identified by two paths $P_1$ and $P_2$ and these paths are well-defined in the sense that they coincide for any member of $[P]$. Moreover, any augmenting path in $[P]$ is obtained by concatenating $P_1$ and $P_2$ and changing the starting time of $P_2$. More precisely, $[P]$ contains only the paths $P|_{P_2}(\theta-\tau_{P_2}), \theta\in [0,T]$, for which $P|_{P_2}(\theta-\tau_{P_2})\in \mathcal{P}_v$. 

We next show that the quotient set $\mathcal{P}_v/\sim$ contains a finite number of equivalence classes. Each equivalence class in $\mathcal{P}_v/\sim$ is defined by two paths $P_1$ and $P2$ with the following properties:
\begin{enumerate}[label = (\theenumi)]
\item\label{it:p1} The first NTP in $P_1$ is $(0,1)$, the last NTP in $P_2$ is $(v,\theta)$ for some $\theta\in [0,T]$, the last two NTPS in $P-1$ are node-linked, and $P_2$ is an arc-path. Moreover, the last NTP in $P_1$ and the first NTP in $P_2$ coincide.  
\item\label{it:p2}  For each maximal arc-subpath~$Q$ of $P_1$ with starting time~$\theta'$ the function $c_{Q}$ has a local minimum at $\theta'$ and is not constant on any open neighborhood containing~$\theta$. 
\item\label{it:p3}  The function $c_{P_1}$ has a local minimum at $\theta'$ and is not constant on any open neighborhood containing~$\theta'$, where $\theta'$ is starting time of $P_1$.
\end{enumerate}
It follows from the last two properties, by a similar argument as in the proof of Lemma~\ref{lem:analytic-nonempty-finite}, that there exists only a finite number of possibilities for $P_1$ and $P_2$. Hence, $\mathcal{P}_v/\sim$ is a finite set.


We now consider an equivalence class $[P]\in \mathcal{P}_v/\sim$, given by two paths $P_1$ and $P_2$. We define a cost function $c_{[P]}:[0,T]\rightarrow \rr$ by
  \begin{align*}
    c_{[P]}(\theta) := \begin{cases}
                                c(P_1) + c_{P_2}(\theta-\tau_{P_2}) & \text{if  $P|_{P_2}(\theta-\tau_{P_2})\in \mathcal{P}_v$}, \\
                                \infty & \text{otherwise}.
                              \end{cases}
  \end{align*}
In other words, $c_{[P]}(\theta):=c(P')$ if there exists a path $P'$ in $[P]$, whose last NTP is $(v,\theta)$; and $c_{[P]}(\theta)=\infty$ if such a path does not exist in $[P]$. Since $P_2$ is an arc-path, $c_{P_2}(\theta-\tau_{P_2})$ is piecewise analytic, so is the cost function $ c_{[P]}$.

Following the above discussion, we can express the cost of each path $P\in \mathcal{P}_v$ in terms of the cost function $c_{[P]}$ as $c(P) = c_{[P]}(\theta)$, where $\theta$ is the last time that we reach node $v$ along~$P$, i.e., $(v,\theta)$ is the last NTP of $P$. It implies that 
\begin{align*}
d_v(\theta) &= \min_{P\in \mathcal{P}_v/\sim}\{c_{[P]}(\theta)\} &&\forall \theta\in [0,T].
\end{align*}
Therefore $d_v$ is piecewise analytic since it is the minimum of a finite number of piecewise analytic functions.
  \qed 
\end{proof}

\section{Optimality conditions and strong duality}
\label{sec:OptCond}

In this section,we turn our attention to the optimality conditions for MCFP. In particular, we first show that not only conditions (CS1)-(CS4) are sufficient for optimality, but also are necessary under Assumptions \ref{asum:rational} and  \ref{asum:analytic}. We then extend the negative cycle optimality condition for MCFP and derive a strong duality result between \eqref{pro:MCFP} and \eqref{pro:MCFP*}. We start with the following lemma which will be used in the proof of the subsequent theorems.

\begin{lemma}
\label{lem:NGOC}
Let $x$ be a piecewise analytic flow. Then $x$ is not optimal if there exists a negative augmenting cycle with respect to $x$.
\end{lemma}
\begin{proof}
Suppose that $W:(v_1,\theta_1),\ldots,(v_q,\theta_q)$ is a negative augmenting cycle under $x$.
We show that there exists a cycle $W'$ with $\ca(W')>0$ and \mbox{$c(W')<0$}.
To prove this, let $Q:(v_\ell,\theta_\ell),\ldots,(v_r,\theta_r)$ be a maximal arc-subpath of $W$. Then we know that there is some $\epsilon>0$ such that $W|_Q(\alpha)$ is also an augmenting cycle for each $\alpha$ in $(\theta_\ell-\epsilon,\theta_\ell)$ or $(\theta_\ell,\theta_\ell+\epsilon)$. We assume without loss of generality that $W|_Q(\alpha)$ is an augmenting cycle for each $\alpha$ in $(\theta_\ell-\epsilon,\theta_\ell)$. Then for some $\alpha\in (\theta_\ell-\epsilon,\theta_\ell)$, we have $|c(Q(\alpha))-c(Q)|<\epsilon$. Here $Q(\alpha)$ denotes the arc path $(v_\ell,v_{\ell+1}),\ldots, (v_{r-1},v_r)$ with starting time $\alpha$. More precisely, we have $Q(\alpha):(v_\ell,\alpha_\ell),\ldots,(v_r,\alpha_r)$ where $\alpha_\ell=\alpha$ and $\alpha_{k+1}(\alpha)=\alpha_{k}+\tau_{v_{k},v_{k+1}}$ for $k=\ell,\ldots,r-1$. Further, $\alpha$ can be chosen in such a way that $\ca(Q(\alpha))>0$. Now we consider the cycle $W|_Q(\alpha)$ and repeat the above procedure for all remaining maximal arc-subpaths of $W$. Let $W'$ be the resulting cycle. It is easy to see that for sufficiently small $\epsilon>0$, we get $\ca(W')>0$ and $c(W')<0$.

Based on the above observation, we assume that $\ca(W)>0$ and \mbox{$c(W)<0$}. For every
$k=2,\ldots,q$ with $e_k=(v_{k-1}, v_k)\in E^r$, there exist $\delta_k, \gamma_k>0$
such that
\begin{align*}
u^r_{e_k}(\theta)&\geq \delta_k && \forall\theta\in [\theta_{k-1},\theta_{k-1}+\gamma_k)~.
\end{align*}

Let
 $\delta$ and $\gamma$ be the minimum of $\delta_k$ and $\gamma_k$,
 respectively, and define $\epsilon_k=2\delta\gamma$.
 Moreover, for every
$k=2,\ldots,q$ with $v_{k-1}=v_k$, there exist $\delta_k,\gamma_k>0$
such that
\begin{align*}
U^r_{v_k}(\theta)&\geq \delta_k && \forall\theta\in (\theta_{k-1}-\gamma_k,\theta_{k}+\gamma_k) &\text{if } \theta_{k-1}<\theta_k~,\\
L^r_{v_k}(\theta)&\geq \delta_k&& \forall \theta\in (\theta_{k}-\gamma_k,\theta_{k-1}+\gamma_k) &\text{if } \theta_{k-1}>\theta_k~.
\end{align*}
Let
$\delta$ and $\gamma$ be the minimum of $\delta_k$ and $\gamma_k$,
respectively. We then define
\begin{align*}
\epsilon_k&:=\begin{cases}
                    2\delta\gamma, & \hbox{if $v_{k-1}\neq v_k$,}\\
                    \delta, & \hbox{if $v_{k-1}=v_k$,}
		\end{cases}
\end{align*}
for $k=2,\ldots,$ and let
\begin{align}
\label{eq:z*}
z^*:=\frac{1}{2\gamma}\min\{\epsilon_2,\epsilon_3,\ldots,\epsilon_{q}\}.
\end{align}
We now define
\begin{align*}
z_{e}(\theta):=\begin{cases}
                   z^* & \hbox{if $e=(v_{k},v_{k+1})\in E, \theta\in  [\theta_{k},\theta_{k}+\gamma)$ for $k=1,\ldots,q-1$,}\\
                    -z^* & \hbox{if $e=(v_{k+1},i=v_{k})\in E,\theta\in [\theta_{k+1},\theta_{k+1}+\gamma)$, for $k=1,\ldots,q-1$,}\\
                    0 & \hbox{otherwise.}
                \end{cases}
\end{align*}
We can easily see that $x+z$ is a feasible flow.

Thus far we have seen that another feasible flow $\bar{x}=x+z$ can be obtained by augmenting a constant flow rate $z^*$, given by \eqref{eq:z*}, along the arcs involved in the cycle
$W$. The \emph{cost of augmenting}, that is, the change in
the objective function value in moving from $x$ to $\bar{x}$, is computed by
 $z^*\sum_{k=2}^{q}{\zeta_k}$, where
\begin{align*}
\zeta_k:=\begin{cases}
                 \int_{\theta_{k}}^{\theta_{k}+\gamma}  {c_{e_k}(\theta)}\,dt & \hbox{if $e_k=(v_{k}v_{k+1})\in E$,}\\
                 \int_{t_{k+1}}^{t_{k+1}+\gamma} {-c_{e_k}(\theta)}\,dt & \hbox{if $e_k=(v_{k+1},v_{k})\in E$,}\\
		 0 & \hbox{otherwise.}
         \end{cases}
\end{align*}
for $k=2,\ldots,q$. Since $z^*>0$, $\bar{x}$ will be a strictly improved feasible solution than
$x(\theta)$ if $\sum_{k=2}^{q}{\zeta_k}<0$. We know that the cost functions $c$ are
piecewise analytic and right-continuous. This implies
$\sum_{k=2}^{q}{\zeta_k}<0$ for $\gamma$ small enough since $c(W)<0$.
This establishes the lemma.
\qed
\end{proof}

We are now in a position to prove the main results of this paper.

\begin{theorem}
\label{thm:RC}
Let $x$ be a piecewise analytic flow.
Then $x$ is optimal for \eqref{pro:MCFP} if and only if there is a piecewise analytic function $\pi$ which is complementary slack with $x$.
\end{theorem}
\begin{proof}
The necessary part of the theorem follows immediately from Theorem ~\ref{thm:CS}. For the sufficiency, suppose that $x$ with corresponding storage $Y$ is optimal. It follows from Lemma~\ref{lem:NGOC} that there are no negative augmenting cycles with respect to $x$. For each node $v\in V$, we consider the function $d_v:[0,T]\rightarrow \rr$, where $d_v(\theta)$ is the cost of a least cost augmenting path from $(s,0)$ to $(v,\theta)$. By Theorem \ref{thm:CDSP-analytic}, the function $d_v$ is well deffined and is piecewise analytic. We now define $\pi_v:=-d_v, v\in V$ and show that the following conditions are met for each arc $e=(v,w)$ and each node $v\in V$ over the time interval $[0,T]$:
\begin{enumerate}[label = (RC\arabic*), leftmargin = *]
\item\label{it:RC1} if $x_e(\theta)>0$, then $c_{e}(\theta)-\pi_v(\theta)+\pi_w(\theta+\tau_e)\leq 0$;
\item\label{it:RC2} if $x_e(\theta)<u_{e}(\theta)$, then $c_{e}(\theta)-\pi_v(\theta)+\pi_w(\theta+\tau_e)\geq 0$;
\item\label{it:RC3} if $Y_v(\theta)>0$ on $(a,b)$, then $\pi_v$ is monotonic decreasing on $(a,b)$;
\item\label{it:RC4} if $Y_v(\theta)<U_v(\theta)$ on $(a,b)$, then $\pi_v$ is monotonic increasing on $(a,b)$.
\end{enumerate}
Suppose by contradiction that the condition \ref{it:RC1} does not hold, that is, there are some arc $e=(v,w)$ and some point in time $\theta$ such that $x_e(\theta)>0$, but $c_{e}(\theta)-\pi_v(\theta)+\pi_w(\theta+\tau_e)<0$ or equivalently $c_{e}(\theta)+d_v(\theta)-d_w(\theta+\tau_e)<0$. Since $x$, $c$, and $d$ are piecewise analytic and thus right-continuous, there is some $\epsilon>0$ for which $x_{e}(\vartheta)>0$ and $c_{e}(\vartheta)+d_v(\vartheta)-d_w(\vartheta+\tau_e)<0$ for each $\vartheta\in [\theta,\theta+\epsilon)$. Let us fix a point $\vartheta\in (\theta,\theta+\epsilon)$ and let $P:(v_1,\theta_1), \ldots, (v_q,\theta_q)$ be a least cost augmenting path from NTP $(0,1)$ to NTP $(v,\vartheta)$. We now consider the augmenting walk $P':(v_1,\theta_1), \ldots, (v_q,\theta_q), (w,\vartheta+\tau_e)$ from NTP $(0,1)$ to NTP $(w,\vartheta+\tau_e)$ with augmenting cost $c(P')=d_v(s)+c_{e}(\vartheta)$. Since $d_w(\vartheta+\tau_e)$ is the cost of the least cost augmenting path from $(s,0)$ to $(w,\vartheta+\tau_e)$ and there are no negative augmenting cycles under $x$, we get $d_w(\vartheta+\tau_e)\leq d_v(\vartheta)+c_{e}(\vartheta)$. This is a contradiction and so the condition \ref{it:RC1} must hold. In a similar way, we can show that the conditions \ref{it:RC2}--\ref{it:RC4} are satisfied.

Obviously, conditions \ref{it:RC1} and \ref{it:RC2} are equivalent to conditions \ref{it:CS1} and \ref{it:CS2}, respectively. We next show that the condition \ref{it:RC3} implies the condition~\ref{it:CS3}. To end this, assume that for some node $v$, the function $\pi_v$ is strictly increasing at some point $\theta$. Since $\pi_v$ is piecewise analytic and consequently of bounded variation, there are two monotonic increasing and piecewise analytic functions $\mu$ and $\eta$ for which $\pi_v=\mu-\eta$. Moreover, since $\pi_v$ is strictly increasing at $\theta$, there is an open interval $(a,b)$ containing $\theta$ such that $\mu$ is strictly increasing at each point in $(a,b)$ and $\eta$ is constant over $(a,b)$. It now follows from condition \ref{it:RC3} that $Y_v$ is identically zero on $(a,b)$. This proves that the condition \ref{it:RC3} must hold. In a similar way, we can show that the condition \ref{it:RC4} implies the condition~\ref{it:CS4}. Therefore, the pair $x$ and $\pi$ satisfies the conditions \ref{it:CS1}--\ref{it:CS4}, establishing the theorem.
\qed\end{proof}

The previous theorem shows that the conditions~\ref{it:RC1}--\ref{it:RC4} are sufficient and necessary for optimality of a given flow $x$.

\begin{theorem}[Negative Cycle Optimality Condition]
\label{thm:NCOC}
Suppose that $x$ is a piecewise analytic flow. Then $x$ is optimal for \eqref{pro:MCFP} if and only if there is no negative augmenting cycle under $x$.
\end{theorem}
\begin{proof}
Having proved Lemma~\ref{lem:NGOC}, it is sufficient to prove that $x$ is optimal if there is no negative augmenting cycle under $x$. Hence, assume that there is no negative cycle with respect to $x$. By a similar argument as in the proof of Theorem~\ref{thm:RC}, we can deduce that there is a piecewise analytic functions $\pi$ which is complementary slack with $x$. It now follows from Theorem~\ref{thm:CS} that~$x$ is optimal for~\eqref{pro:MCFP}.
\qed
\end{proof}

\begin{theorem}[Strong Duality]
\label{StrongDuality1}
There exist piecewise analytic optimal solutions $x$ and $\pi$ for \eqref{pro:MCFP} and \eqref{pro:MCFP*}, respectively, so that $V[\text{MCFP},x]=V[\text{MCFP}^*,\pi]$.
\end{theorem}
\begin{proof}
We know from Theorem \ref{thm:analytic-flow} that \eqref{pro:MCFP} has a piecewise analytic optimal solution~$x$. Then, it follows from Theorem~\ref{thm:RC} that there exists a piecewise analytic function~$\pi$, which is complementary slack with $x$. The assertion now follows from Theorem~\ref{thm:CS}.
\qed
\end{proof}

\section{Conclusions}
\label{sec:con}
In this paper, we have studied the minimum cost flow problem in a time-varying network to capture temporal features of many real-world problems. In this problem, arc and node costs, arc and node capacities, and supplies and demands can change over time where time is modeled as a continuum. We developed several network-related optimality conditions under the assumption that the input functions are piecewise analytic and the transit times are constant and rational. These results can be used to develop algorithms for solving the minimum cost flow over time problem in a similar way as in static network flows.  For example, Theorem \ref{thm:NCOC} lays the ground for an algorithmic approach like the negative cycle-canceling algorithm for the static minimum cost flow problem. This algorithm maintains a feasible solution at each iteration and successively improves the solution towards optimality. More specifically, the algorithm first establishes an initial feasible solution. It then proceeds by identifying negative augmenting cycles and sending flow along these cycles, while preserving feasibility. The algorithm terminates when the network contains no negative cycle. Theorem \ref{thm:NCOC} implies that when the algorithm terminates it has found an optimal solution. Yet, we have to investigate how to implement the essential steps of such an algorithm. Further details are beyond the scope of this paper and are left for further work.

\section*{Acknowledgment}
We would like to thank two anonymous referees for many useful comments that helped to improve the presentation of the paper. We are grateful to Martin Skutella for valuable suggestions that enhanced this paper. 

\bibliographystyle{plain}
\bibliography{mybib}

\appendix
\section{MCFP with discontinuous cost functions}
\label{sec:appendix}

Throughout Section~\ref{sex:SAP}, for simplicity of notation and clarity of presentation, we have assumed that the cost functions are continuous. This condition ensures that the cost function $c_Q$ is continuous and hence attains it minimum at some point on a closed interval. We have used this fact in Lemma \ref{lem:path-exists-sink} to show that there exists a least cost augmenting path in $\mathcal{P}_{\loc}$. Here, we show that all results in Sections~\ref{sex:SAP} and \ref{sec:OptCond} hold for the case in which cost functions have some discontinuities. To do so, we require to define the cost of an augmenting path and cycle in a different way.

Let $P:(v_1,\theta_{1}),\ldots,(v_q,\theta_{q})$ be an augmenting path from NTP $(s,1)$ to NTP $(t,T)$. We observe that for $k=1,\ldots,q-1$ the following conditions hold:
\begin{enumerate}[label = (\roman*)]
\item\label{it:AugPath1} if $e_k=(v_{k},v_{k+1})\in E^r$, then $u^r_{e_k}$ is not identically zero on any open interval containing $\theta_k$,
\item\label{it:AugPath3} if $v_{k}=v_{k+1}$ and $\theta_{k}<\theta_{k+1}$, then $U^r_{v_{k}}(\theta)>0$ for each $\theta\in (\theta_{k},\theta_{k+1})$,
\item\label{it:AugPath4} if $v_{k}=v_{k+1}$ and $\theta_{k+1}<\theta_{k}$, then $L^r_{v_{k}}(\theta)>0$ for each $\theta\in (\theta_{k+1},\theta_{k})$.
\end{enumerate}
In particular, if $Q~:~(v_\ell,\theta_\ell),\ldots,(v_r,\theta_r)$ is a maximal arc-subpath of $P$, then $P|_Q(\alpha)$ is also an augmenting path for each $\alpha$ in $(\theta_\ell-\epsilon,\theta_\ell)$ or $(\theta_\ell,\theta_\ell+\epsilon)$ for some sufficiently small $\epsilon>0$. Depending on whether  $P|_Q(\alpha)$ is an augmenting path for each $\alpha$ in $(\theta_\ell-\epsilon,\theta_\ell)$, $(\theta_\ell,\theta_\ell+\epsilon)$, or $(\theta_\ell-\epsilon,\theta_\ell+\epsilon)$, we define the cost $c'(Q)$ of $Q$ as
\begin{align}
\label{eq:aug-cost'}
c'(Q)&:=\begin{cases}	
                 \sum_{k=\ell}^{r-1}{c_{e_k}(\theta_{k}-)}& \hbox{if $\alpha$ in $(\theta_\ell-\epsilon,\theta_\ell)$~,}\\
                 \sum_{k=\ell}^{r-1}{c_{e_k}(\theta_{k}+)}& \hbox{if $\alpha$ in $(\theta_\ell,\theta_\ell+\epsilon)$~,}\\
                 \sum_{k=\ell}^{r-1}{\min\{c_{e_k}(\theta_{k}-),c_{e_k}(\theta_{k}+)\}} & \hbox{if $\alpha$ in $(\theta_\ell-\epsilon,\theta_\ell+\epsilon)$~,}
                \end{cases}
\end{align}
where $e_k=(v_k,v_{k+1})$. Notice that $c_{e_k}(\theta_{k}-)$ and $c_{e_k}(\theta_{k}+)$ denote the limit of~$c_{e_k}$ at $\theta_k$ from the left and from the right, respectively, i.e.,
\begin{align*}
c_{e_k}(\theta_{k}-):=\lim_{t\rightarrow \theta_k^-} c_{e_k}(\theta)&& \text{and} && c_{e_k}(\theta_{k}+):=\lim_{t\rightarrow \theta_k^+}c_{e_k}(\theta).
\end{align*}

We now define the \emph{cost} $c'(P)$ of $P$ as $c'(P):=\sum_{Q}c'(Q)$,
where the sum is taken over all maximal arc-subpaths $Q$ of $P$.
Notice that $c'(P)$ is equal to $c(P)$, given by\eqref{eq:cost}, for the case that the cost functions are continuous, but in general it is not the case as shown in the following example.

\begin{example}
\label{ex:AugPath}

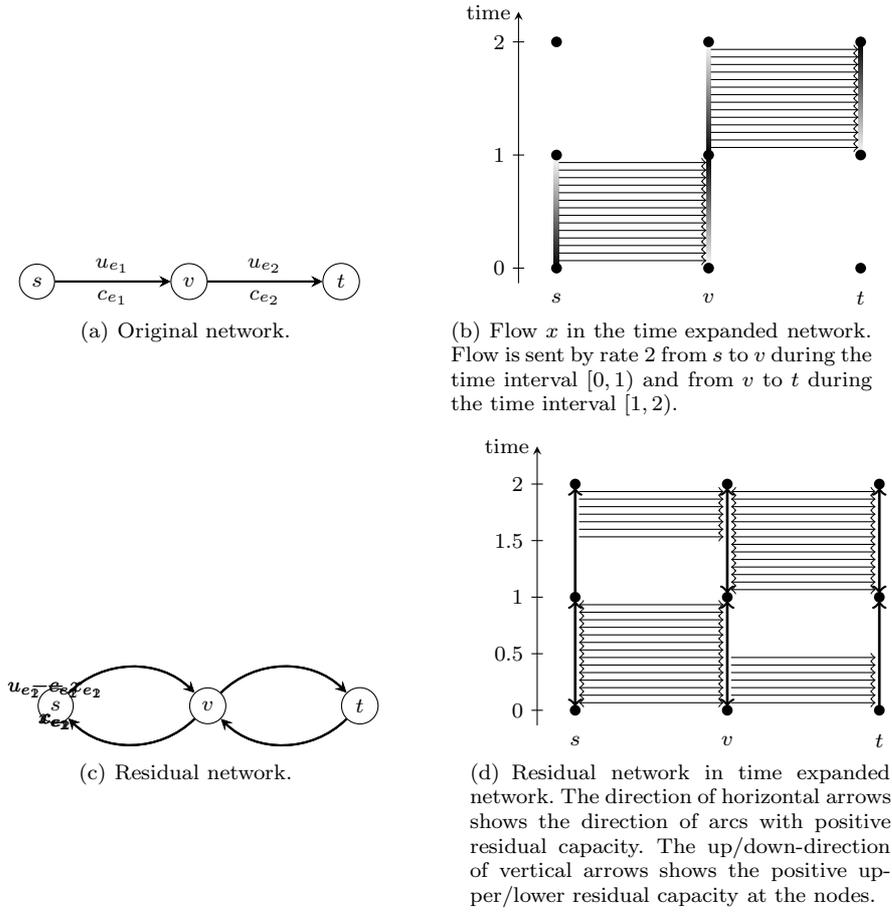
\begin{figure}
    \centering
    \subfigure[\label{fig:ex1Orig}Original network.]{
     \begin{tikzpicture}
      \tikzstyle{node}=[circle, draw]
      \tikzstyle{arc} = [draw, thick, -stealth]
      \foreach \pos/\name in {{(0cm,0cm)/$s$}, {(2cm,0cm)/$v$}, {(4cm,0cm)/$t$}} {
        \node[node] (\name) at \pos {\name};      }
      \foreach \tail / \head / \cap /\cost in {{$s$/$v$/$u_{e_1}$/$c_{e_1}$}, {$v$/$t$/$u_{e_2}$/$c_{e_2}$}} {
        \path[arc] (\tail) --  (\head) node[above,text centered,midway]{\cap}; \path[arc] (\tail) --  (\head) node[below,text centered,midway]{\cost};       }
    \end{tikzpicture}}
        \hspace{1cm}
    \subfigure[\label{fig:ex1Expa}Flow $x$ in the time expanded network. Flow is sent by rate $2$ from $s$ to $v$ during the time interval $[0,1)$ and from $v$ to $t$ during the time interval $[1,2)$. ]{
    \begin{tikzpicture}
      \tikzstyle{node}=[circle, fill, inner sep =.05cm]
      \tikzstyle{arc} = [draw, thick, -stealth]
      \path[draw, -stealth] (-.5cm,-.2cm) --node[at end,auto]{time} (-.5cm,3.4cm);
      \path[draw,] (-.575cm,0cm) -- (-.425cm,0cm) node[at start, anchor=east] {0};
      \path[draw,] (-.575cm,1.5cm) -- (-.425cm,1.5cm) node[at start, anchor=east] {1};
      \path[draw,] (-.575cm,3cm) -- (-.425cm,3cm) node[at start, anchor=east] {2};
      \foreach \pos/\name in {{(0cm,-0.4cm)/$s$}, {(2cm,-0.4cm)/$v$}, {(4cm,-0.4cm)/$t$}} {\node (\name) at \pos {\name};}
      \shade[top color=white,bottom color=white,middle color=white] (0,0) rectangle (2,1.5);
      \shade[top color=black,bottom color=white,middle color=gray] (1.97,0) rectangle (2.03,1.5);
      \shade[top color=white,bottom color=white,middle color=white] (2,1.5) rectangle (4,3);
      \shade[top color=black,bottom color=white,middle color=gray] (3.97,1.5) rectangle (4.03,3);
      \shade[top color=white,bottom color=black,middle color=gray] (1.97,1.5) rectangle (2.03,3);
      \shade[top color=white,bottom color=black,middle color=gray] (-.03,0) rectangle (0.03,1.5);
      \foreach \pos/\name in {{(0cm,0cm)/1}, {(0cm,1.5cm)/2}, {(0cm,3cm)/3}, {(2cm,0cm)/4}, {(2cm,1.5cm)/5}, {(2cm,3cm)/6},
      {(4cm,0cm)/7}, {(4cm,1.5cm)/8}, {(4cm,3cm)/9}} {\node[node] (\name) at \pos {};}
      \foreach \x in {1,...,14}{\draw[black,->] (0.03cm,0.1*\x cm) -- (1.97cm,0.1*\x cm) node[at start, anchor=east] {};}
      \foreach \x in {16,...,29}{\draw[black,->] (2.03cm,0.1*\x cm) -- (3.97cm,0.1*\x cm) node[at start, anchor=east] {};}
     \end{tikzpicture}}
    \subfigure[\label{fig:ex1Back}Residual network.]{
    \begin{tikzpicture}
        \tikzstyle{node}=[circle, draw]
      \tikzstyle{arc} = [draw, thick, -stealth]
      \foreach \pos/\name in {{(0cm,0cm)/$s$}, {(2cm,0cm)/$v$}, {(4cm,0cm)/$t$}} {
        \node[node] (\name) at \pos {\name};      }
      \foreach \tail / \head / \cap / \cost in {{$s$/$v$/$u_{e_1}-x_{e_1}$/$c_{e_1}$}, {$v$/$t$/$u_{e_2}-x_{e_2}$/$c_{e_2}$}} {
        \path[arc] (\tail) to [bend left=45]  (\head) node[above,text centered,midway]{\cap}; \path[arc] (\tail) to [bend left=45]  (\head) node[below,text centered,midway]{\cost};}
      \foreach \tail / \head / \cap / \cost in {{$s$/$v$/$x_{e_1}$/$-c_{e_1}$}, {$v$/$t$/$x_{e_2}$/$-c_{e_2}$}} {
        \path[arc] (\head) to [bend left=45]  (\tail) node[above,text centered,midway]{\cost}; \path[arc] (\head) to [bend left=45]  (\tail) node[below,text centered,midway]{\cap};}
          \end{tikzpicture}}
       \hspace{1cm}
    \subfigure[\label{fig:ex1ResExpa}Residual network in time expanded network. The direction of horizontal arrows shows the direction of arcs with positive residual capacity. The up/down-direction of vertical arrows shows the positive upper/lower residual capacity at the nodes.]{
    \begin{tikzpicture}
      \tikzstyle{node}=[circle, fill, inner sep =.05cm]
      \tikzstyle{arc} = [draw, thick, -stealth]
      \path[draw, -stealth] (-.5cm,-.2cm) --node[at end,auto]{time} (-.5cm,3.5cm);
      \path[draw,] (-.575cm,0cm) -- (-.425cm,0cm) node[at start, anchor=east] {0};
      \path[draw,] (-.575cm,1.5cm) -- (-.425cm,1.5cm) node[at start, anchor=east] {1};
      \path[draw,] (-.575cm,3cm) -- (-.425cm,3cm) node[at start, anchor=east] {2};
      \path[draw,] (-.575cm,0.75cm) -- (-.425cm,0.75cm) node[at start, anchor=east] {0.5};
      \path[draw,] (-.575cm,2.25cm) -- (-.425cm,2.25cm) node[at start, anchor=east] {1.5};
      \foreach \pos/\name in {{(0cm,-0.4cm)/$s$}, {(2cm,-0.4cm)/$v$}, {(4cm,-0.4cm)/$t$}} {\node (\name) at \pos {\name};}
      \foreach \pos/\name in {{(0cm,0cm)/1}, {(0cm,1.5cm)/2}, {(0cm,3cm)/3}, {(2cm,0cm)/4}, {(2cm,1.5cm)/5}, {(2cm,3cm)/6},
      {(4cm,0cm)/7}, {(4cm,1.5cm)/8}, {(4cm,3cm)/9}} {\node[node] (\name) at \pos {};}
      \foreach \x in {1,...,14}{\draw[black,<->] (0.05cm,0.1*\x cm) -- (1.95cm,0.1*\x cm);}
      \foreach \x in {1,...,7}{\draw[black,->] (2.05cm,0.1*\x cm) -- (3.95cm,0.1*\x cm);}
      \foreach \x in {23,...,29}{\draw[black,->] (0.05cm,0.1*\x cm) -- (1.95cm,0.1*\x cm);}
      \foreach \x in {16,...,29}{\draw[black,<->] (2.05cm,0.1*\x cm) -- (3.95cm,0.1*\x cm);}
      \path[draw,,line width=1pt,<->] (0,0.05) -- (0,1.45);
      \path[draw,,line width=1pt,->] (0,1.55) -- (0,2.95);
      \path[draw,,line width=1pt,<->] (2,0.05) -- (2,1.45);
      \path[draw,,line width=1pt,<->] (2,1.55) -- (2,2.95);
      \path[draw,,line width=1pt,->] (4,0.05) -- (4,1.45);
      \path[draw,,line width=1pt,<->] (4,1.55) -- (4,2.95);
     \end{tikzpicture}}
\caption{\label{fig:AugPath}Network for Example~\ref{ex:AugPath}.}
\end{figure}

We consider the network shown in Fig.~\ref{fig:ex1Orig}.
The transit costs and transit capacities are as follows:
\begin{align*}
c_{e_1}(\theta)&=\begin{cases}
            1&\hbox{$0\leq t< 0.5$,}\\
            2&\hbox{$0.5\leq t< 1$,}\\
            1&\hbox{$1\leq t\leq 2$,}
            \end{cases}&
c_{e_2}(\theta)&=\begin{cases}
            1&\hbox{$0\leq t< 1$,}\\
            2&\hbox{$1\leq t< 1.5$,}\\
            1&\hbox{$1.5\leq t\leq 2$,}
            \end{cases}\\
u_{e_1}(\theta)&=\begin{cases}
            4t&\hbox{$0\leq t< 1$,}\\
            0&\hbox{$1\leq t< 1.5$,}\\
            1&\hbox{$1.5\leq t\leq 2$,}
            \end{cases}&
u_{e_2}(\theta)&=\begin{cases}
            1&0\leq t< 0.5,\\
            0&0.5\leq t< 1,\\
            4t-2&1\leq t\leq 2.
            \end{cases}
\end{align*}
The storage capacities are given as
\begin{align*}
U_s(\theta)=\infty,&& U_v(\theta)=1,&& U_t(\theta)=\infty,&& \theta\in [0,2].
\end{align*}
The transit times and storage costs are assumed to be zero. The problem is to send an initial storage of one
unit from node 1 to node 3 within the time interval~$[0,2]$. One possible solution $x$ is obtained as follows. We send flow into arc $e_1$ with rate $2\theta$ within the interval $[0,1)$. The flow arriving at node $2$ is stored there till time $1$. So there will be one unit of flow at node $2$ at time $1$. We send this amount of flow into arc $e_2$ with rate $2\theta-2$ within the interval $[1,2]$. Fig.~\ref{fig:ex1Expa} shows the flow~$x$ in the corresponding time-expanded network.
Formally, $x$ is given by
\begin{align*}
x_{e_1}(\theta)&=\begin{cases}
            2\theta&\theta\in [0,1),\\
            0&\theta\in [1,2),
            \end{cases}&
x_{e_2}(\theta)&=\begin{cases}
            0&\theta\in [0,1),\\
            2\theta-2&\theta\in [1,2],
            \end{cases}
\end{align*}
with corresponding storage
\begin{align*}
Y_{s}(\theta)&=\begin{cases}
            1-\theta^2&\theta\in [0,1),\\
            0&\theta\in [1,2),
            \end{cases}&
Y_{v}(\theta)&=\begin{cases}
            \theta^2&\theta\in [0,1),\\
            2\theta-\theta^2&\theta\in [1,2),
            \end{cases}&
Y_{w}(\theta)&=\begin{cases}
            0&\theta\in [0,1),\\
            (\theta-1)^2&\theta\in [1,2).
            \end{cases}
\end{align*}

We are now interested in identifying the augmenting paths and augmenting cycles. Fig.~\ref{fig:ex1Back} depicts the network with backward arcs and Fig.~\ref{fig:ex1ResExpa} depicts the paths and cycles with positive residual capacities in the corresponding time-expanded network. However, there are more augmenting paths and cycles in addition to those shown in Fig.~\ref{fig:ex1ResExpa}, whose residual capacities are zero. Some of them are given below
\begin{align*}
P_1:&(s,0),(2,0),(3,0),(3,2),\\
P_2:&(s,0),(2,0),(2,0.5),(3,0.5),(3,2),\\
P_3:&(s,0),(1,0.5),(2,0.5),(3,0.5)(3,2),\\
P_4:&(s,0),(1,1.5),(2,1.5),(2,2),(3,2),\\
P_5:&(s,0),(1,2),(2,2),(3,2)\\
W_1:&(s,0),(1,2),(2,2),(3,2),(3,1),(2,1),(2,0),(s,0),\\
W_2:&(1,1),(1,2),(2,2),(3,2),(3,1.5),(2,1.5),(2,1),(1,1),\\
W_3:&(1,0.5),(1,2),(2,2),(3,2),(3,1.5),(2,1.5),(2,0.5),(1,0.5),
\end{align*}
with costs
\begin{align*}
c(P_1)=2,	&& c(P_2)=2,		&&c(P_3)=2,	&& c(P_4)=2,		&&c(P_5)=2,\\
c'(P_1)=2,	&& c'(P_2)=1,	&&c'(P_3)=2,	&& c'(P_4)=2,	 &&c'(P_5)=2,
\end{align*}
and
\begin{align*}
c(W_1)=-1,	&& c(W_2)=-2,		&&c(W_3)=-2,\\
c'(W_1)=-1,	&& c'(W_2)=1,	&&c'(W_3)=-1.
\end{align*}
We observe that the equality $c(P)=c'(P)$ does not hold for some path or cycle $P$.

As mentioned already above, an augmenting path (or cycle) must satisfy the conditions \eqref{it:AugPath1}-\eqref{it:AugPath4}. But the other direction may not hold, that is,
a path satisfying these conditions is not necessarily an augmenting path in general. The following paths and cycles show this fact:
\begin{align*}
P_6:&(s,0),(2,0),(2,1),(3,1),(3,2),\\
P_7:&(s,0),(1,1),(2,1),(3,1),(3,2),\\
W_4:&(1,1),(1,2),(2,2),(3,2),(3,1),(2,1),(1,1).\\
\end{align*}

\end{example}

Now let $Q~:(v_\ell,\theta_\ell),\ldots,(v_r,\theta_r)$ be an arc-subpath of $P$. We have observed that there exists an (inclusion-wise) maximal closed interval $[a,b]$, say, containing $\theta_\alpha$ so that the path $P|_Q(\alpha)$, given by \eqref{eq:P(alpha)}, is an augmenting path for every $\alpha\in [a,b]$. We define a cost function $c'_Q:[a,b]\rightarrow \rr$ with respect to $Q$ as
\begin{align*}
c'_Q(\alpha)&:=\begin{cases}
        	\sum_{k}^{\ell}{c_{e_{k}}(\alpha_{k}-)}& \alpha=v,\\
		\sum_{k}^{\ell}\min\{c_{e_{k}}(\alpha_{k}-),c_{e_{k}}(\alpha_{k}+)\}& \alpha\in (u,v),\\
		\sum_{k}^{\ell}{c_{e_{k}}(\alpha_{k}+)}& \alpha=u,
             \end{cases}
\end{align*}
where $e_k=(v_k,v_{k+1})$. We recall that $\alpha_\ell=\alpha$ and $\alpha_{k+1}(\alpha)=\alpha_{k}+\tau_{e_{e_k}}$ for $k=\ell,\ldots,r-1$. For the case that $u=v=t_\ell$, we define $c'_Q(\alpha):=c'(Q)$, where $c'(Q)$ is given by \eqref{eq:aug-cost'}. The function $c'_Q$ is lower semi-continuous at any point $\alpha\in [a,b]$ and such a function attains its local minimum on a closed interval. This fact shows that Theorem~\ref{thm:CDSP-existence} and \ref{thm:CDSP-analytic} hold.

We next proceed to show that the results in Section~\ref{sec:OptCond} are still valid. We consider a feasible flow $x$ and suppose that $W: (v_1,\theta_1), \ldots, (v_q,\theta_q)$ is an augmenting cycle. We have defined the cost of $W$ in two different ways: $c'(W)$ in terms of the cost of arc-subpaths of $W$ and $c(W)$ as the sum of the costs of the arcs at the times they appear around the cycle $W$. Further, we have observed that these two values are not equal in general. However, we have the following result taht provides another characterization of augmenting cycles with negative cost.

\begin{lemma}
\label{lem:neg-aug-cycle}
Let $x$ be a piecewise analytic flow. The network $G$ contains a negative augmenting cycle if and only if there is a cycle $W$ with $\ca(W)>0$ and \mbox{$c'(W)<0$}.
\end{lemma}
\begin{proof}
Suppose first that $W$ is a cycle with $\ca(W)>0$ and $c'(W)<0$. Clearly $W$ is an augmenting cycle since $\ca(W)>0$. So we need to show that $c(W)<0$. Recall that $c(W)=\sum_{Q} c(Q)$ where sum is taken over all maximal arc-subpaths $Q$ of~$W$ and $c(Q)$ is computed by \eqref{eq:cost}. For each maximal arc-subpath $Q:(v_\ell,\theta_\ell),\ldots,(v_r,\theta_r)$ of $W$, we define $c'(Q):=\sum_{k:(v_k,v_{k+1})\in E^r} c_{v_k,v_{k+1}}(\theta_{k})$ where the index $k$ varies from~$\ell$ to $r-1$. The fact that $\ca(W)>0$ implies that there exists some $\epsilon>0$ so that $W|_Q(\alpha)$ is an augmenting cycle for each $\alpha\in (\theta_\ell,\theta_\ell+\epsilon)$. Due to the definition of $c(Q)$ and $c'(Q)$ and the fact that cost functions are right-continuous, we can conclude $c(Q)\leq c'(Q)$. Therefore, $c(W)\leq c'(W)$ which gives the result in one direction.

To prove the other direction, suppose that $W$ is a negative augmenting cycle. Let $Q:(v_\ell,\theta_\ell),\ldots,(v_r,\theta_r)$ be a maximal arc-subpath of $W$. Then we know that there is some $\epsilon>0$ such that $W|_Q(\alpha)$ is also an augmenting cycle for each $\alpha$ in $(\theta_\ell-\epsilon,\theta_\ell)$ or $(\theta_\ell,\theta_\ell+\epsilon)$. We assume without loss of generality that $W|_Q(\alpha)$ is an augmenting cycle for each $\alpha$ in $(\theta_\ell-\epsilon,\theta_\ell)$. Then for some $\alpha\in (\theta_\ell-\epsilon,\theta_\ell)$, we have $|c'(Q(\alpha))-c(Q)|<\epsilon$. Here $Q(\alpha)$ denotes the arc path $(v_\ell,v_{\ell+1}),\ldots, (v_{r-1},v_r)$ with starting time $\alpha$. More precisely, we have $Q(\alpha):(v_\ell,\alpha_\ell),\ldots,(v_r,\alpha_r)$ where $\alpha_\ell=\alpha$ and $\alpha_{k+1}(\alpha)=\alpha_{k}+\pi_{v_{k},v_{k+1}}$ for $k=\ell,\ldots,r-1$.
Further, $\alpha$ can be chosen in such a way that $\ca(Q(\alpha))>0$. Now we consider the cycle $W|_Q(\alpha)$ and repeat the above procedure for all remaining maximal arc-subpaths of $W$. Let $W'$ be the resulting cycle. It is easy to see that for sufficiently small $\epsilon>0$, we get $\ca(W')>0$ and $c'(W')<0$.
\qed
\end{proof}

This lemma  implies that all results in Section~\ref{sec:OptCond} remain valid even if the cost functions have some discontinuities.

\end{document}